\documentclass[11pt]{article}
\usepackage[colorlinks]{hyperref}
\hypersetup{linkcolor=cyan,filecolor=cyan,citecolor=cyan,urlcolor=cyan}
\usepackage{xspace}
\usepackage{thm-restate,color,xcolor}
\usepackage{amsmath,amssymb,bbm,amsthm,algorithmicx,algorithm,algpseudocode}
\usepackage{fullpage}
\usepackage{boxedminipage}
\usepackage{algorithm}
\usepackage{amsmath}
\usepackage{mathtools}
\usepackage{framed}
\usepackage[framemethod=tikz]{mdframed}
\usepackage{titlesec}
\usepackage{lipsum}%
\usepackage{cleveref,aliascnt}
\usepackage{tikz}
\usepackage{wrapfig}

\algblock{ParFor}{EndParFor}
\algnewcommand\algorithmicparfor{\textbf{parfor}}
\algnewcommand\algorithmicpardo{\textbf{do}}
\algnewcommand\algorithmicendparfor{\textbf{end\ parfor}}
\algrenewtext{ParFor}[1]{\algorithmicparfor\ #1\ \algorithmicpardo}
\algrenewtext{EndParFor}{\algorithmicendparfor}

\usepackage[framemethod=tikz]{mdframed}
\usepackage{tikz}
\usetikzlibrary {shapes.symbols}

\usetikzlibrary{shapes.geometric}
\usetikzlibrary{arrows}
\usetikzlibrary{arrows.meta}
\usetikzlibrary{patterns}
\usetikzlibrary{shapes.misc}

\newcommand{\dist}{\mbox{\rm dist}}

\newtheorem{theorem}{Theorem}[section]
\newtheorem{lemma}[theorem]{Lemma}
\newtheorem{meta-theorem}[theorem]{Meta-Theorem}

\newtheorem{proposition}[theorem]{Proposition}
\newtheorem{observation}[theorem]{Observation}
\newtheorem{definition}[theorem]{Definition}

\newcommand{\poly}{\operatorname{\text{{\rm poly}}}}

\usepackage{geometry}
\begin{document}
\date{}

\title{Connectivity Certificate against Bounded-Degree Faults: \\ Simpler, Better and Supporting Vertex Faults}

\author{
Merav Parter\thanks{Weizmann Institute. Email: \texttt{merav.parter@weizmann.ac.il}. Supported by the European Research Council (ERC) under the European Union’s Horizon 2020 research and
innovation programme, grant agreement No. 949083.}
\and Elad Tzalik\thanks{Weizmann Institute. Email: \texttt{elad.tzalik@weizmann.ac.il}. Supported by the Adams Fellowship Program of the Israel Academy of Sciences and Humanities.} }

\maketitle

\begin{abstract}
An $f$-edge (or vertex) connectivity certificate is a sparse subgraph that maintains connectivity under the failure of at most $f$ edges (or vertices). It is well known that any $n$-vertex graph admits an $f$-edge (or vertex) connectivity certificate with $\Theta(f n)$ edges (Nagamochi and Ibaraki, Algorithmica 1992). 
A recent work by (Bodwin, Haeupler and Parter, SODA 2024) introduced a new and considerably stronger variant of connectivity certificates that can preserve connectivity under any failing set of edges with bounded degree. For every $n$-vertex graph $G=(V,E)$ and a degree threshold $f$, an $f$-Edge-Faulty-Degree (EFD) certificate is a subgraph $H \subseteq G$ with the following guarantee: For any subset $F \subseteq E$ with $\deg(F)\leq f$ and every pair $u,v \in V$, $u$ and $v$ are connected in $H - F$ iff they are connected in $G - F$. For example, a 
$1$-EFD certificate preserves connectivity under the failing of any matching edge set $F$ (hence, possibly $|F|=\Theta(n)$). In their work, [BHP'24] presented an expander-based approach (e.g., using the tools of expander decomposition and expander routing) for computing $f$-EFD certificates with $O(f n \cdot \poly(\log n))$ edges. They also provided a lower bound of $\Omega(f n\cdot \log_f n)$, hence $\Omega(n\log n)$ for $f=O(1)$.

In this work, we settle the optimal existential size bounds for $f$-EFD certificates (up to constant factors), and also extend it to support vertex failures with bounded degrees (where each vertex is incident to at most $f$ faulty vertices). Specifically, we show that for every $n>f/2$, any $n$-vertex graph admits an $f$-EFD (and $f$-VFD) certificates with $O(f n \cdot \log(n/f))$ edges and that this bound is tight. Our upper bound arguments are considerably simpler compared to prior work, do not use expanders, and only exploit the basic structure of bounded degree edge and vertex cuts. 
\end{abstract}




\section{Introduction} 
Connectivity certificates are fundamental fault-tolerant graph structures that preserve the graph connectivity under edge (or vertex) failures. Formally, for a given $n$-vertex graph $G$, a subgraph $H \subseteq G$ is an $f$-edge connectivity certificate if for every subset $F$ of at most $f$ edges (or vertices) in $G$, the connected components of $H - F$ are equivalent to that of $G - F$. Since their introduction by Nagamochi and Ibaraki \cite{NagamochiI92} connectivity certificates have played a key role in various computational settings \cite{Matula93, CheriyanKT93, KargerM97, GhaffariK13, ForsterNYSY20, LiNPSY21}. The size bounds of connectivity certificates are well-understood. \cite{NagamochiI92} demonstrated that every $n$-vertex graph admits an $f$-edge (or vertex) connectivity certificate with $O(f n)$ edges. It is easy to show that this bound is also tight by considering some $(f+1)$-regular $n$-vertex graph $G$: any edge $(u,v)$ must be part of the output $f$-edge certificate of $G$ since upon failing of all $f$ other edge incident to $u$, the only surviving edge incident to $u$ is $(u,v)$. The locality of this lower bound argument raises the question of whether one can tolerate a considerably larger number of faults, with roughly the same certificate size, provided that one restricts only the number of faulty edges incident to any given vertex.

Bodwin, Haeupler and Parter \cite{BodwinHP24} recently answered this question in the affirmative. They introduced a new notion of connectivity certificates whose size bound depends on the degree of the failing edge set, rather than the total number of faults. For a subset of faulty edges $F \subseteq G$, the \emph{faulty-degree} $\deg(F)$ is the largest number of faults in $F$ incident to any given vertex. For example, for a matching set $F$, we have that $\deg(F)=1$ while $|F|$ might be as large as $n/2$. 
The connectivity certificates that arise under this new local model of fault tolerance are denoted as faulty-degree (FD) certificates, formally defined as follows:

\begin{definition} [Edge-Faulty-Degree (EFD) Connectivity Certificates,\cite{BodwinHP24}]
Given a graph $G$, an edge-subgraph $H$ is called an $f$-EFD connectivity certificate if, for any $F \subseteq E$ such that $\deg_F(x) \le f$ for all vertices $x$, any two vertices $u,v$ are connected in $H - F$ iff there they are connected in $G - F$.
\end{definition}

It is easy to see that an $f$-EFD connectivity certificate is stronger than the standard notion of $f$-edge certificate and in fact, can tolerate linearly (in $n$) more faults. Despite the stronger model, 
\cite{BodwinHP24} showed that $f$-EFD connectivity certificates have nearly the same size complexity as their weaker $f$-edge certificate counterpart. The algorithm for computing these certificates is based on the tools of expander decomposition and expander routing, which have been also used in the past in the context of dynamic graph sparsification \cite{BernsteinBGNSS022,Brand0PKLGSS23,ChenKLPGS23}. Due to the useful properties of expanders, the output certificates obtained in this approach have, in fact, a stronger property of recovering not just connectivity, but also shortest path distances up to a $\text{polylog } n$ factor, in the presence of bounded degree faults. On the lower bound side, \cite{BodwinHP24} showed that $\Omega(n\log_f n)$ edges are necessary. Hence, at least for $f=O(1)$, there is a logarithmic separation in the size of $f$-edge certificates vs.\ $f$-EFD certificates. One limitation of the expander machinery of \cite{BodwinHP24} is that it is based on edge rather than vertex-expanders, and hence it is unclear how to extend their approach to support also vertex failures $F \subseteq V$ of bounded degree (i.e., where each vertex is incident to at most $f$ vertices in $F$).\footnote{While the tools of expander decomposition and routing exist for vertex expanders, see \cite{LongS22}, the vertex expander decomposition cannot guarantee that each vertex expander has large minimum degree, which was important for the construction of \cite{BodwinHP24}.}

\paragraph{Related Notion: Faulty-Degree Spanners.} A closely related notion to FD connectivity certificates is the notion of FD spanners. FD spanners allow one to recover approximate shortest paths, instead of mere connectivity. Formally, a subgraph $H \subseteq G$ is a $f$-EFD $t$-spanner if for every $u,v \in V$ and $f$-degree set $F \subseteq E$, it holds that $\dist_{H - F}(u,v)\leq t \cdot \dist_{G - F}(u,v)$. Note that an $f$-EFD $t$-spanner for any finite $t$ is also an $f$-EFD certificate. Indeed, constructions of fault-tolerant spanners (in various faulty settings) have been used to provide sparse certificates, in a black-box manner \cite{Par19,Parter22,PetruschkaST24} by simply setting $t=O(\log n)$. By employing the well-known blocking-set method of \cite{BP19,BDRSODA22}, \cite{BodwinHP24} provided a size analysis of a suited greedy spanner algorithm, yielding the following existential bound:  For all positive integers $n, k, f$, every $n$-vertex undirected weighted graph has a $f$-EFD $(2k-1)$-spanner on at most $f^{1 - 1/k} n^{1+1/k} \cdot O(k)^{k}$ edges.
Due to the extra $O(k)^{k}$ factor in the size bound, this only yields $f$-EFD certificates with suboptimal sparsity of $fn^{1+o(1)}$ edges.

\subsection{Our Contribution}
In this work, we strengthen both the upper and lower bounds for $f$-EFD certificates obtained by \cite{BodwinHP24} and settle the optimal existential size bounds of these structures. Our construction is shown to hold also in a mixed fault-tolerant setting where the faulty set $F$ is a subset of $V(G)\cup E(G)$ such that each vertex is incident to at most $f$ failed vertices and edges in $F$. Lastly and more importantly, our algorithms are considerably simpler, do not use the expander machinery, and our analysis only exploits basic properties of bounded-degree cuts.

To present our results in their most generality we follow the recent terminology of \cite{PetruschkaST24} that provides a unified language to capture both vertex and edge failures, at once. 

\smallskip
\noindent \textbf{New: Faulty-Degree-Mixed (MFD) Certificates.} The degree of a faulty-set $F \subseteq V(G)\cup E(G)$ in a graph $G$ denoted by $\deg_G(F)$ is the largest number of edge and vertex faults in $F$ that are incident to any given vertex. For a vertex $v \in V$ with neighbor set $N_G(v)$ in $G$, let $\deg_G(v,F)=\{u \in N_G(v) ~\mid~ u \in F \mbox{~or~} (u,v)\in F\}$. Then, $\deg_G(F)$ is defined by $\deg_G(F)=\max_{v\in V(G)-F} \deg_G(v,F)$.
An $f$-Mixed-Faulty-Degree (MFD) set $F$ in $G$ is a set $F \subset V(G)\cup E(G)$ with $\deg_G(F)\leq f$.

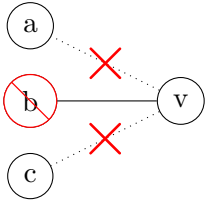
\begin{figure} [h]
\begin{center}
\begin{tikzpicture}
    \node [draw, circle] (v) at (2,0) {v};
    \node [draw, circle] (b) at (0,0) {b};
    \node [draw, circle] (a) at (0,1) {a};
    \node [draw, circle] (c) at (0,-1) {c};
    \node [correct forbidden sign,line width=0.1ex,draw=red,fill=white] {b};
    \draw (v) -- (b);
    \draw[dotted] (v) -- node [midway, red] {\Huge \bf $\times$} (a);
    \draw[dotted] (v) -- node [red] {\Huge \bf $\times$} (c);

\end{tikzpicture}
\end{center}
\caption{Shown is a $3$-MFD set $F=\{(v,a),(v,c),b\}$, where vertex $v$ is incident to three faults.}
\end{figure}

\begin{definition}[MFD Connectivity Certificates]\label{def:MFD}
    For a given graph $G=(V,E)$, a subgraph $H\subseteq G$ is an $f$-MFD certificate if for every $f$-MFD set $F$ in $G$ (i.e., with $\deg_G(F)\leq f$), $G- F$ and $H - F$ have the same connected components. When restricting $F$ to be a subset of edges (resp., vertices), the output certificate is denoted as $f$-EFD (resp., $f$-VFD) connectivity certificate.
\end{definition}

Our key result is a simple upper bound for $f$-MFD connectivity certificates:

\begin{theorem}[Main Result]\label{thm:size-fMFD}
    For all positive integers $n>f$, every simple $n$-vertex graph has an $f$-MFD connectivity certificate with $O(\min\{fn \log (n/f),n^2\})$ edges.
\end{theorem}
Note that for $f=\Omega(n)$, there is an immediate lower bound of $\Omega(n^2)$ which holds already for the weaker notion of $f$-vertex connectivity certificates (that assume that the \emph{total} number of vertex faults is bounded by $f$). A major advantage of this result is that our certificates are constructed by a very natural and simple greedy algorithm, which is the obvious
extension of the standard greedy algorithm used by \cite{BodwinHP24} to build $f$-EFD spanners. Our proof technique is completely different and considerably simpler than that of \cite{BodwinHP24}. While our primary objective is in settling the optimal bounds of $f$-EFD certificates, we remark that our algorithm is inefficient in terms of running time. 

We complement this by a matching lower bound for $f$-EFD connectivity certificate. Similarly to \cite{BodwinHP24} the lower bound graph is obtained by a simple extension of the hypercube graph. The construction of \cite{BodwinHP24} considers a hypercube with alphabet of size $f$. Our alternative simple modification of the hypercube provides the following stronger (and tight) bound:

\begin{theorem}[Lower Bound for EFD Certificates]\label{thm:lb-fMFD}
For every $n>f \geq 1$, there exists an $n$-vertex graph $G^*=(V,E)$ such that any $f$-EFD connectivity certificate $H$ for $G^*$ has $\Omega(fn \log (n/f))$ edges. 
\end{theorem}
This should be compared with the bound of $\Omega(fn (\log n/\log f))$ by \cite{BodwinHP24}, and is $\Omega(\log(n))$ sharper in the regime of polynomially (in $n$) many faults. Our lower bound for EFD certificates provides a logarithmic separation in the size of $f$-edge certificates vs.\ $f$-EFD certificates for any $f=n^{1-\epsilon}$ for any constant $\epsilon \in (0,1)$. 

\subsection{Our Approach}
Our approach is based on a new cut-based extension of the blocking set method by Bodwin and Patel \cite{BP19} that has been used in a large collection of the previous works on fault tolerant spanners \cite{BP19,DR20, BodwinDR21,BodwinDR22,BodwinHP24,PetruschkaST24}. In what follows, we first provide a quick introduction to the blocking set method and then explain its extension to faulty-degree certificates, which yields a considerably simpler size arguments compared to its recent application in \cite{BodwinHP24}.

\paragraph{The Blocking-Set Method.} 
The blocking set method, introduced by Bodwin and Patel~\cite{BP19} was originally designed to analyze the size of $f$-VFT $(2k-1)$ spanner\footnote{An $f$-VFT $(2k-1)$-spanner is a subgraph $H \subseteq G$ satisfying that $\dist_{H - F}(u,v)\leq (2k-1)\dist_{G -F}(u,v)$ for every $u,v \in V$ and at most $f$ vertex faults $F \subseteq V$.} $H$ obtained by the output of the (exponential-time) \emph{FT greedy algorithm} presented by Bodwin, Dinitz, Parter and Vassilevska Williams \cite{BDPW18}. In their ingenious simplification of the original analysis by \cite{BDPW18}, Bodwin and Patel~\cite{BP19} defined the notion of a blocking set: A collection of edge-vertex pairs that ``block" all short  (of length $\leq 2k$)  cycles in $H$, in the sense that each short cycle in $H$ contains one of the pairs. The size analysis recipe takes two main steps. First, it is shown that the output greedy spanner $H$ has a small blocking set, specifically, in the context of $f$-VFT spanners, of size $f |E(H)|$. The second part of the argument (which is usually the technically involved one) shows that a graph that admits a small blocking set is close to having high girth, forcing it to be sparse (by standard Moore bounds). To prove this statement, all prior applications of the blocking set method, applied a probabilistic argument that bounds the size of a randomly sampled subgraph of $H$.

Blocking sets have been also employed in the context of $f$-EFD spanners in \cite{BodwinHP24}. A key challenge that arises in the bounded faulty degree setting is that a graph with a small FD blocking set might still include many short simple cycles.
\cite{BodwinHP24} mitigated that challenge by bounding the number of short cycles of a \emph{specific} and delicate form. This complication led to an extra (possibly unnecessary) exponential overhead of $k^k$ in the size bound, which makes it less applicable in the context of FD connectivity certificates, where one needs to block arbitrarily long cycles.

\paragraph{New: Blocking Bounded-Degree Cuts for MFD Certificates.} 
Our simple idea is to diverge from the high-girth based approach and take a cut-based approach. The key insight is that with each edge $e$ added by the greedy \emph{certificate} algorithm, the fault set $F_e$ that caused $e$ to be added \emph{forms a cut} in the current certificate $H_e$ (i.e., the subgraph that the greedy algorithm maintained up to deciding on $e$). We then crucially use the fact that all cycles whose last edge is $e$ are blocked by $e\times F_e$ \emph{simultaneously}. 
This shift from cycles to cuts perspective helps controlling the space of all cycles in $H$. In contrast to the random sampling approach used before, our size bound follows by a simple a divide-and-conquer approach that only exploits the basic properties of bounded degree faults. 



\paragraph{Computational Aspects.} The complexity of our greedy algorithm (say, for $f$-EFD certificates) can be determined by the time to compute $m$ instances of the Min-Max $s$-$t$ cut problem \cite{CharikarGS17}. The latter asks  
for finding an $s$-$t$ cut minimizing the largest number of cut edges incident on any vertex. This problem is a special case of max-min correlation clustering in weighted complete graphs. Charikar, Gupta, and Schwartz \cite{CharikarGS17} presented a $O(\sqrt{n})$-approximation algorithm for this problem. It is easy to see that any $\alpha$-approximation solution for this problem yields a polynomial time algorithm for computing $f$-EFD certificates with $O(\alpha \cdot f \cdot n \log(n/f))$ edges, hence this approach calls for improved approximation bounds for the Min-Max $s$-$t$ cut problem.

   %
To provide a poly-time implementation of their greedy EFD-spanner algorithm, \cite{BodwinHP24} approximated a \emph{length-restricted} variant of this Min-Max $s$-$t$ cut problem, which unfortunately yields a linear approximation in our length-unbounded setting. It will be very interesting to find a polynomial time implementation of our greedy algorithm. It is noteworthy that in the (standard) fault-tolerant setting, the first exponential-time greedy algorithm by \cite{BDPW18} which obtained existentially optimal $f$-VFT spanners (for fixed stretch $k$), led to the existentially optimal (and still exponential-time) VFT spanners of \cite{BP19}, which eventually led to algorithms that construct optimal $f$-VFT spanners in a time-efficient manner \cite{DR20, BodwinDR21}. We are hopeful that this work will have a similar effect.




\section{Upper Bound for Mixed Faults}

We now describe the FT greedy certificate algorithm in the MFD setting. 
Throughout, we fix some ordering of the edges $E=\{e_1,\ldots, e_m\}$. We use the notation $e_i=\{u_i,v_i\}$ to denote the endpoints of an edge. Denote by $G_{\leq i}$, the graph $(V,\{e_1,\ldots,e_i\})$ and by $G_{<i}$ the graph $(V,\{e_1,\ldots,e_{i-1}\})$.



\begin{definition}[Damaging an Edge]
    A set $F$ damages $e$, if $e \not\in G-F$.
\end{definition}

We now describe the FT greedy certificate algorithm in the MFD setting.

\begin{algorithm}[H]
\caption{$\mathsf{FTGreedyMFDCertificate}(G,f)$}\label{alg:mfd-greedy-certificate}

\begin{algorithmic}[1]
    \State $H \gets (V, \emptyset)$
    \For{$i=1,2,\ldots,m$}
        \If{There exists $F$ with $\deg_G(F) \leq f$ that does not damage $e_i$, such that  $u_i$ and $v_i$ are disconnected in $H-F$}
        \State $H \gets H \cup \{e_i\}$. 
        \EndIf
    \EndFor
    \State \Return $H$.
\end{algorithmic}
\end{algorithm}

A natural approach to analyze the size of the FT-certificate produced by the greedy MFD-certificate algorithm is to notice that with each edge $e_i$ inserted to the output certificate $H$, one can find a ''blame set'' $F_i \subseteq V \cup E$, which caused $e_i$ to be included into $H$ with $F_i$ of bounded degree $f$ in $H$ at the time of considering $e_i$. Moreover, the set of blocks $B = \cup_{e\in H} e_i\times F_i$ satisfies that every cycle $C$ has a block $b \in B$ whose two coordinates are in $C$. Then, one (hopefully) proves that any graph $H$ that admits a set $B$ as above that blocks all cycles must be sparse (and hence so does the output of the greedy certificate algorithm). Unfortunately, similar perhaps to the challenges encountered in \cite{BodwinHP24}, we were unable to make progress in this direction, and therefore take a cut-based approach.

To analyze this greedy certificate algorithm we proceed with defining the notion of connectivity blocking sets. 
\begin{definition}[Connectivity Blocking Sets]\label{def:blocking_set}
     Given a graph $G = (V, E)$ with ordered edges $E=\{e_1,\ldots, e_m\}$, a connectivity blocking set $B$ for $G$ is a collection of pairs $\{(e_i,F_i)\}_{i=1}^{m}$ with $e_i\in E$, $ F_i \subseteq V \cup E$ such that: (1) $F_i$ doesn't damage $e_i$, and (2) $u_i$ and $v_i$ are disconnected in $G_{<i}-F_i$.
\end{definition}

We observe that a connectivity blocking set ``blocks" all the cycles in $G$:

\begin{observation}
   For every cycle $C$ in $G$, there exists a block $(e_i, F_i)$ such that $e_i \in C$ and $F_i \cap C \neq \emptyset$.   
\end{observation}

\begin{proof}
    Fix $C \subseteq G$ and let $e_i$ be the latest edge in $C$, i.e., $C - \{e_i\} \subseteq G_{<i}$. Let $F_i$ be the set associated with $e_i$, namely, that $(e_i,F_i)$ is in the blocking set. By the definition of $F_i$, we have that $u_i$ and $v_i$ are not connected in $G_{<i} - F_i$, hence $F_i$ intersects all the cycles that contain $e_i$ in $G_{\leq i}$ and $C$ in particular. Since $C$ was arbitrary the claim follows.
\end{proof}

We now specialize this definition to the bounded degree setting:

\begin{definition}[MFD Connectivity Blocking Set]\label{def:mfd_blocking_set}
     A connectivity blocking set for $G$, $B=\{(e_i,F_i)\}_{i=1}^{m}$ is an $f$-MFD connectivity blocking set if $\deg_{G}(F_i)\leq f$ for all $i$.
\end{definition}

This definition is equivalent to the definition of the faulty-degree (FD) blocking set by 
\cite{BodwinHP24} (see Def. 6.3 therein) with the only distinction that their definition blocks all short cycles (of length at most $k$) and in our case, we block all cycles (hence, $k=\infty$). Since we take a cut-based analysis, we explicitly define our blocking sets in terms of cuts rather than blocking of cycles. From this point, we refer to MFD connectivity blocking set by  MFD blocking sets, for short. To ease notation we define $F^V=F \cap V$ and $F^E = F \cap E$ to be the vertex and edge parts of $F$.

One useful property of MFD blocking sets as defined above is that they are \emph{hereditary} as formalized in the following lemma. This property plays a key role in our divide and conquer approach for bounding the size of the output subgraph.


\begin{lemma}[MFD blocking set is hereditary]\label{lem:hereditary}
    If $G =(V,E)$ has an $f$-MFD blocking set $\left\{(e_i,F_i)\right\}_i$, then for $S \subseteq V$ the induced graph $G[S]=(S,E_S)$ has $\left\{\left(e_i, \left(F_i^V\cap S\right) \cup \left( F_i^E \cap \binom{S}{2} \right) \right) \right\}_{e_i \in G[S]}$ as an 
    $f$-MFD blocking set (with the edge order induced from $G$). 
\end{lemma}

\begin{proof}
    It is clear that $\deg_{G[S]}(F_i) \leq  \deg_{G}(F_i)$, and that if $F_i$ doesn't damage an edge $e_i \in G[S]$ then the same is true for any subset of $F_i$ hence it remains to prove that the induced blocking satisfies \Cref{def:blocking_set}(2).
    
 To show this, 
 notice that $G_{< i}[S] - \left(  \left(F_i^V\cap S\right) \cup \left(F_i^E \cap \binom{S}{2} \right) \right) $ is a subgraph of $G_{<i} - F_i$ hence since the latter has $u_i$ and $v_i$ in different components, so does the former, hence $u_i$ and $v_i$ are disconnected in $G_{< i}[S] - \left(\left(F_i^E \cap \binom{S}{2}\right) \cup \left(F_i^V \cap S \right) \right)$ which implies that the induced graph with the induced blocking set satisfies \Cref{def:blocking_set}(2). 
\end{proof}

We now observe that by construction, the output of the FT-greedy algorithm, $H$, can be accompanied by a blocking set following the algorithm's execution.

\begin{observation}
    Let $H$ be the output of \Cref{alg:mfd-greedy-certificate} with the edge ordering induced from $G$. For each $e \in H$ denote $F_e$ the set of faults $F$ which caused $e$ to be added to $H$ (line 3 of \Cref{alg:mfd-greedy-certificate}). Then $\{ (e,F_e) \}_{e\in H}$ is an $f$-MFD blocking blocking set for $H$.
\end{observation}

\begin{proof}
    Whenever an edge $e$ is added, the condition in line $3$ in \Cref{alg:mfd-greedy-certificate} must be satisfied, which guarantees $\deg_{G}(F_e) \leq f$ and $F_e$ doesn't damage $e$. Condition (2) of \Cref{def:blocking_set} holds since the edges are considered by \Cref{alg:mfd-greedy-certificate} in order.
\end{proof}

For the size analysis, we have the following lemma:

\begin{lemma}
    Let $G$ be a simple graph on $n$ vertices and $m$ edges, that admits an $f$-MFD blocking set $B=\{e_i,F_i\}$ with $\deg_{G}(F_i) \leq f$, then $m=O(fn\log(\frac{n}{f}))$. In particular the certificate produced by \Cref{alg:mfd-greedy-certificate} is of size $O(fn\log(\frac{n}{f}))$\footnote{Unless stated otherwise, $\log$ is taken with base $2$.} assuming $n>10f$\footnote{If $n=O(f)$ then $m=O(n^2)$ is best possible for simple $G$ as mentioned in the introduction.}.
\end{lemma}

\begin{proof}
    
    To prove the claim we show that $G$ has an orientation of the edges with outdegree $\leq 10 f\log(\frac{n+2f}{f})$ which concludes the claim. The proof is by induction on $n$. The claim is clear for $n\leq 10f$ as any graph on $n$ vertices has an orientation with outdegree $n$.
    
    Let $\{(e_i,F_i)\}_{i=1}^m$ denote the $f$-MFD blocking set guaranteed by the assumption. To ease the notation, we denote $F_i \cup \{e_i\}$ by $F^{*}_i$.

    Let $e_m$ be the edge of largest index in $G$. We then have that $F_m$ has bounded degree at most $f$ in $G$ and therefore the set $F_m^*$ is of degree bounded by $f+1$ in $G$. 
    We have that $G-F_m^* = G_{\leq m} - F_m^* = G_{<m}-F_m$ and by \Cref{def:blocking_set}(2) $G_{<m}-F_m$ is disconnected, and hence so is $G-F_m^*$. Notice that $F_i^*$ is only guaranteed to cut the graph $G_{\leq i}$ hence it may possibly be that only with $i=m$ the set $F_i^*$ forms a cut in $G=G_{\leq m}$ (when $i=m$ the edge $e_i$ is the \emph{last} edge).
    
    Using $F_m^*$ to cut $G$, we may partition $V$ into two non-empty subsets $L, R$ such that (1) each component of $G-F_m^*$ is contained in either $L$ or $R$ and (2) $F_m^V$ is fully contained in either $L$ or $R$. Observe that \emph{non-empty} subsets $L,R$ exist, by using the non-damaging property of the blocking set with $(e_m,F_m)$.  From now on we assume that w.l.o.g $|L|\leq \frac{n}{2}$ and that $F_m \subseteq R$ (if $F_m^V \subseteq L$ take $L'= L-F_m^V$ and $R'= R \cup F_m^V$ as the partition).  
    
    We now proceed to orient the edges of $G$. By \Cref{lem:hereditary}, $G[L],G\left[R\right]$ admit $f$-MFD blocking sets, hence by the induction hypothesis, we can orient $G[R]$ with outdegree $\leq 10f\log(\frac{(n-1)+2f}{f})$, and orient $G[L]$ with outdegree at most $ 10f\log(\frac{n/2+2f}{f})$. This is the crucial point where we needed the hereditary property, even though $G[L],G\left[R\right]$ do not contain the last edge $e_m$ they contain a ''relative'' last edge induced from the MFD blocking set of $G$. I.e., we continue the argument recursively  on $G[L]$ (resp., $G[R]$) using the last edge $e_\ell$ in $G[L]$ (resp., $e_r$ in $G[R]$). 

    Finally, orient all edges that are not already oriented to go from $L$ to $R$. Explicitly, the edges that need an orientation are edges lost while removing $F_m^*$ which is a set of bounded degree at most $f+1$, hence in this way we increased the outdegree of the vertices in $L$ by at most $f+1$.
    To conclude the proof, since the outdegree of vertices in $R$ satisfies the induction hypothesis we need to bound the outdegree of the vertices of $L$ in the orientation. To do so we observe that:

    \[ 10f\log\left(\frac{n/2+2f}{f} \right) + (f+1) \leq  10f\log\left(\frac{n/2+2f}{f} \right) + 2f \leq 10f \left(\log\left(\frac{n/2+2f}{f} \right)+\frac{1}{5}\right).  \]
    Recall we assumed $f \leq \frac{n}{10}$ and since $\frac{1}{5} \leq \log(1.2)$ the outdegree is bounded by:

    \[ 10f \left(\log\left(1.2 \cdot \frac{n/2+2 (n/10)}{f} \right)\right) \leq 10f \log\left(\frac{n}{f}\right) \leq 10f \log\left(\frac{n+2f}{f}\right) .\]
    As required.
\end{proof}

We now prove that $H$ is indeed a certificate concluding the proof of the main theorem.

\begin{proposition}
    The output subgraph $H$ of \Cref{alg:mfd-greedy-certificate} is an $f$-MFD certificate of $G$.
\end{proposition}

\begin{proof}
It is sufficient to show that for every edge $e_i=(u_i,v_i) \in G - H$ and for every $F \subseteq E(G) \cup V(G)$ such that $\deg_G(F) \leq f$ and $e$ is not damaged by $F$, it holds that $u_i$ and $v_i$ are connected in $H -F$. Since $e_i \notin H$, and $\deg_G(F) \leq f$ we must have that $u_i$ and $v_i$ are connected in $H_{<i}-F$. Since $H_{<i}\subseteq H$, it also holds that $u_i$ and $v_i$ are connected in $H -F$, as desired.
\end{proof}

\section{Lower Bound for $f$-EFD Certificates}

Finally, turn to consider the proof of Thm. \ref{thm:lb-fMFD}


\begin{proof}[Proof of \ref{thm:lb-fMFD}]
    The lower bound graph is given by a slight modification to the hypercube graph. Let $G$ be the hypercube graph on $\lfloor n/f \rfloor$ vertices. We define an $n$-vertex graph $G_f$ by replacing each vertex $v \in V(G)$ by $f$ copies $v_1,\ldots v_f$. To describe the edges of $G_f$, connect $v_i$ and $u_j$ for every $i,j \in \{1,\ldots, f\}$ if $v,u$ are neighbors in the original hypercube $G$. In other words, replace each $G$-edge with an $f \times f$ biclique, see \Cref{fig:efd_lower_bound}. 
    
    We claim that $G_f$ has no proper $f$-EFD connectivity certificate as a strict subgraph. Notice that in the hypercube $G$, each edge $e \in E(G)$ belongs to some perfect matching $M_{e}$ whose removal disconnects the graph. For each $e \in E(G)$, we can then define the edge-set $M_e^f = \{ (v_{i},u_{j}) \in V(G_f) \mid \{v,u\} \in M_{e}, i,j \in \{1,\ldots, f\}\}$. Note that $M_e^f$ has bounded degree $f$ in $G_f$, and its removal leaves $G_f$ disconnected. 
    
    Fix an edge $e'=(v_i,u_j)$ in $G_f$ and let $e=(v,u)$ be the corresponding edge in $G$. Our goal is to show that $G_f -\{e'\}$ is not a valid $f$-EFD certificate for $G_f$. 
    Assume towards a contradiction that $G_f -\{e'\}$ is a $f$-EFD certificate. Consider now the failing set $M'=M_e^f \setminus \{e'\}$ in $G_f$ and observe that $e' \in M_e^f$.  
    Since $M_e^f$ is a cut in $G_f$, it holds that $M'$ is a cut in $G_f -\{e'\}$, and in particular $v_i$ and $u_j$ are not connected in $G_f -\{e'\}$. In contrast, $v_i$ and $u_j$ are connected in $G_f-M'$. Leading to a contradiction that $G_f -\{e'\}$ is a $f$-EFD certificate. We conclude that $G_f$ has no proper $f$-EFD connectivity certificate as a strict subgraph. The claim follows as number of edges $G_f$ is $\Omega(f^2 \cdot n/f \cdot \log(n/f) )$, as required.
       
    

\end{proof}

\begin{figure}[h]
\centerline{\includegraphics[scale=.45]{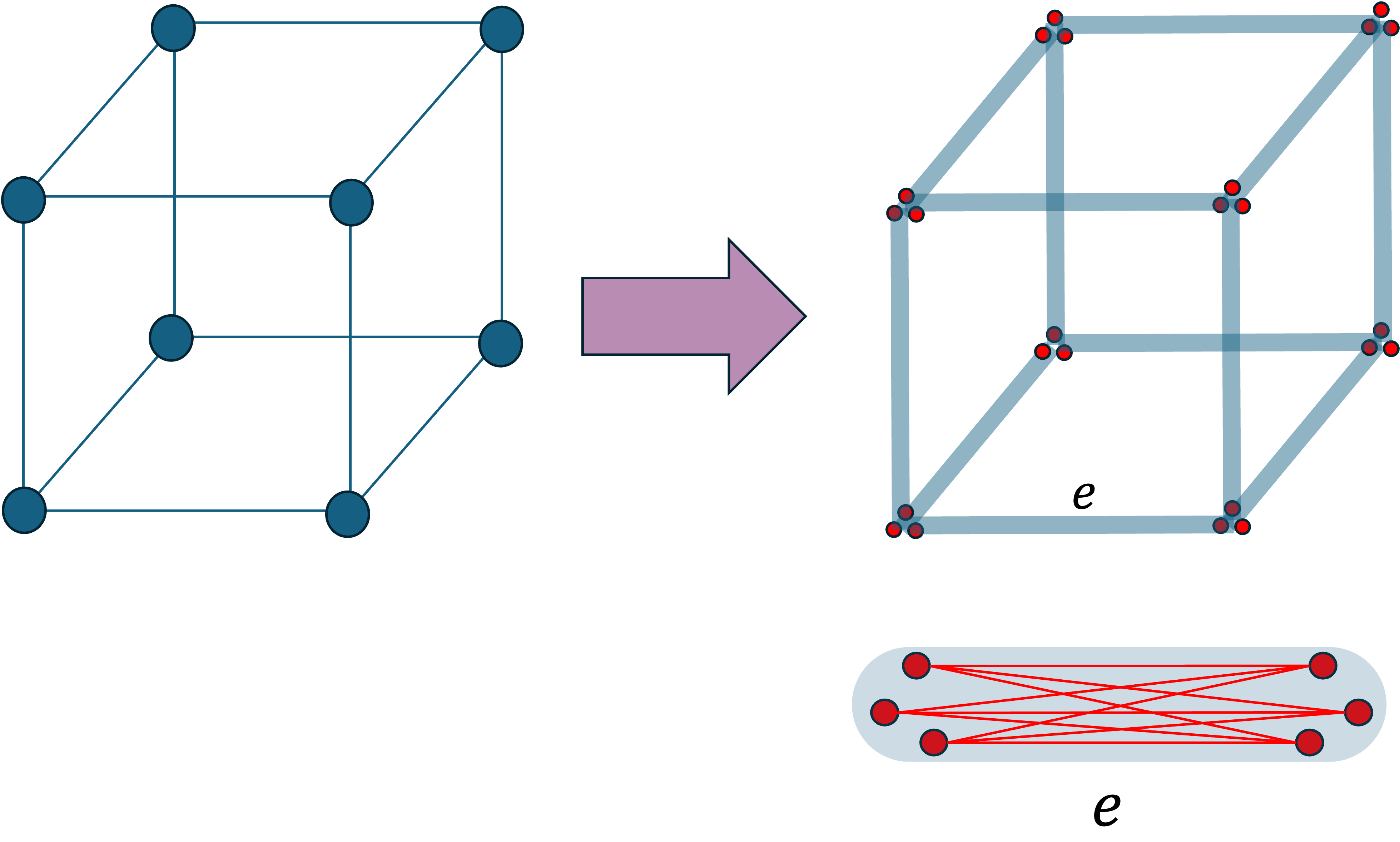}}
\caption{The lower bound graph is formed by taking the hypercube and replacing each vertex replaced by $f$ vertices, and each edge by an $f\times f$ biclique. In the figure $f=3$, every transparent edge $e$ in the right cube represents a $3 \times 3$ biclique.}
\label{fig:efd_lower_bound}
\end{figure}

\paragraph{A Remark on Multigraphs.} We next consider the setting of multigraphs and show that these graphs admit $f$-MFD certificates with $O(fn\log n)$ edges, and that this bound is existentially tight. The setting of multigraphs differs from simple graph in the case of edge failures. In the VFD setting, one can assume w.l.o.g., that $G$ is simple, as parallel edges must fail together in presence of vertex faults.

The upper bound proof for simple graph breaks in the base of the induction, where we used the fact that the maximum degree is $n$. It is easy to see that for our purposes, one can assume, w.l.o.g., that any edge has multiplicity of at most $f+1$, and hence the maximum degree is $O(fn)$. The upper bound argument then starts by assuming that any graph with $f$-MFD blocking set has $O(fn\log n)$ edges, with the base case of constant $n$ following from the bound of $f+1$ parallel edges between each neighboring pair. The rest of the proof is identical. 

For the lower bound part, consider the hypercube with $n$ vertices, and replace each edge with $f$ parallel edges. This gives a matching lower bound of $O(fn\log n)$ edge. We therefore have:

\begin{theorem}
    Any $n$-vertex multigraph $G=(V,E)$ admits an $f$-MFD certificate with $O(fn\log n)$ edges, and this bound is existentially tight. 
\end{theorem}

\paragraph{Open Ends.} In this work we settle the existential optimal size bounds for $f$-EFD certificates by analyzing a simple greedy algorithm. Our work leaves several interesting open ends. While our construction provides an upper bound of $O(fn(\log(n/f))$ edges for $f$-VFD certificates, we are currently missing a matching lower bound. The only known lower bound for the latter is given by the $\Omega(fn)$ lower bound known for the weaker notion of $f$-vertex certificates. Closing this gap will be very interesting. It will be also interesting to simplify the blocking set methodology of $f$-EFD spanners by \cite{BodwinHP24} and settling the optimal sparsity of these structures. Finally, as noted in the introduction, we hope that our work will inspire future time-efficient implementations of our greedy approach. 

\bibliographystyle{alpha}
\bibliography{refs}

\end{document}